\documentclass[a4paper]{lmcs}
\pdfoutput=1

\usepackage{lastpage}
\lmcsdoi{15}{4}{11}
\lmcsheading{}{\pageref{LastPage}}{}{}%
{Apr.~19,~2018}{Nov.~27,~2019}{}

\keywords{algebraic language theory, infinite tree, $\omega$-clone}

\usepackage[all]{xy}
\usepackage{amsmath}
\usepackage{amssymb}
\usepackage{color}
\usepackage{graphicx}
\usepackage{xspace}

\def\imagetop#1{\vtop{\null\hbox{#1}}}

\begin{document}

\title[A non-regular language of infinite trees]{
	A non-regular language of infinite trees \\ that is recognized by a sort-wise finite algebra
}
\titlecomment{Research supported by the European Research Council (ERC) under the European Union's Horizon 2020 research and innovation programme (ERC consolidator grant LIPA, agreement no. 683080).}

\author{Miko\l{}aj Boja\'nczyk}
\address{University of Warsaw}
\email{bojan@mimuw.edu.pl}

\author{Bartek Klin}
\address{University of Warsaw}
\email{klin@mimuw.edu.pl}

\begin{abstract}
$\omega$-clones are multi-sorted structures that naturally emerge as algebras for infinite trees, just as $\omega$-semigroups are convenient algebras for infinite words. In the algebraic theory of languages, one hopes that a language is regular if and only if it is recognized by an algebra that is finite in some simple sense. We show that, for infinite trees, the situation is not so simple: there exists an $\omega$-clone that is finite on every sort and finitely generated, but recognizes a non-regular language.
\end{abstract}

\maketitle

\newcommand{\trans}{\mathsf{prof}}
\newcommand{\M}{\mathsf M}
\newcommand{\clone}{\mathsf C}
\newcommand{\clonereg}{\mathsf C^{\mathrm{reg}}}
\newcommand{\rank}[1]{\mathrm{rank}(#1)}
\newcommand{\diva}{{\sc{(da)}}}
\newcommand{\homf}{{\sc{(hf)}}}
\newcommand{\fing}{{\sc{(fg)}}}
\newcommand{\eqdef}{\stackrel{\text{\tiny def}}=}
\newcommand{\trees}{{\mathsf{trees}}}
\newcommand{\terms}[1]{{\mathsf T}\!_{#1}}
\newcommand{\powerset}{{\mathsf P}}
\newcommand{\powersetfin}{{\powerset_{\textrm{fin}}}}
\newcommand{\muddles}{{\mathsf M}}
\newcommand{\unit}{\mathsf{unit}}
\newcommand{\flatt}{\mathsf{flat}}
\newcommand{\aalg}{\mathbf{A}}
\newcommand{\balg}{\mathbf{B}}
\newcommand{\Nat}{\mathbb N}
\newcommand{\hsp}{{\sc hsp}\xspace}
\newcommand{\mso}{{\sc mso}\xspace}
\newcommand{\hs}{{\sc hs \xspace}}
\newcommand{\aut}{\mathcal A}
\newcommand{\nodes}{\mathsf{nodes}}
\newcommand{\set}[1]{\{#1\}}
\newcommand{\dom}{\mathsf{dom}}
\newcommand{\game}{\mathsf G}
\newcommand{\mult}{\mathsf{pr}}
\newcommand{\parfun}{\rightharpoonup}
\newcommand{\monad}{\mathsf M}
\newcommand{\facto}{\mathsf F}
\newcommand{\poly}[2]{\mathsf{pol}_{#1}#2}

\newcounter{ourexamplecounter}
\newenvironment{myexample}{
\medskip

\refstepcounter{ourexamplecounter}
\smallskip\noindent{\textbf{{Example \arabic{ourexamplecounter}. }}}}{
$\Box$ \smallskip 
}

\newcounter{runningcounter}
\newenvironment{running}{
\medskip

\refstepcounter{runningcounter}
\smallskip\noindent{\textbf{{Running Example \arabic{runningcounter}. }}}}{
$\Box$ \smallskip 
}


\section{Introduction}


The central theme in the algebraic theory of languages is that regular languages are exactly those which are recognized by finite algebras. Here, ``regular'' means definable in monadic second order logic (\mso), or recognizable by a suitable kind of finite automata. This theme occurs for many kinds of objects.
Important examples are: finite words (here the algebras are semigroups), $\omega$-words (here the algebras are $\omega$-semigroups, see~\cite{perrin_pin_words}), scattered countable linear orders (here the algebras are $\Diamond$-semigroups, see~\cite{DBLP:journals/ijfcs/RispalC05}) or unrestricted countable linear orders (here the algebras are $\circ$-semigroups, see~\cite{shelah_composition,DBLP:conf/icalp/CartonCP11}). 

In all of these cases the languages recognized by algebras with a finite universe are exactly those that can be defined in \mso. This fact is remarkable especially for the infinite structures, because then the multiplication operation of an algebra is usually infinitary and hence it may be an infinite object, even when the universe of the algebra is finite. The implication ``if a language is recognized by a finite algebra then it is definable in \mso'' uses a subtle interplay between the finiteness of the universe and the associativity of a multiplication operation in an algebra, which can be exploited together with regularity results like Ramsey's Theorem (used implicitly already by B\"uchi and more explicitly in~\cite{DBLP:conf/icalp/Wilke91,perrin_pin_words} for $\omega$-words), Hausdorff's theorem on scattered linear orders used in~\cite{DBLP:journals/ijfcs/RispalC05}, or Simon's Factorisation Forest Theorem used in~\cite{DBLP:conf/icalp/CartonCP11}. The equivalence of recognizability by finite algebras and \mso definability carries over to other finite structures, such as finite ranked trees~\cite{ThaWri68} or finite graphs of bounded treewidth~\cite{bojanczyk2016definability}.

A missing piece of this puzzle is the equivalence of algebraic recognizability and \mso for infinite trees. The hard part is proving that  if a language of infinite trees is recognized by a finite algebra, then it is definable in \mso. A subproblem is defining what a ``finite algebra'' actually is. 

So far there have been three approaches to this question~\cite{DBLP:conf/concur/BojanczykI09,DBLP:journals/tcs/Blumensath13a,DBLP:journals/mst/IdziaszekSB16}, each one only partially successful. 

The solution proposed in~\cite{DBLP:conf/concur/BojanczykI09} is a two-sorted algebra that intuitively represents regular trees (i.e. ones that have only finitely many non-isomorphic subtrees) with zero or one free variable. The chosen notion of finiteness is that the universe is finite on each of the two sorts, and the multiplication operation is \mso-definable. This approach has two disadvantages: (a) a design choice of the algebra is that it only represents regular trees and not arbitrary infinite trees; and (b) the definition of a ``finite algebra''  uses the \mso logic, and hence a correspondence of such algebras with \mso is not  so surprising. 

Another algebra for infinite trees was proposed in~\cite{DBLP:journals/mst/IdziaszekSB16}. There the notion of a finite algebra used purely structural conditions, but it modeled only trees with countably many branches. Such trees cannot feature recursive branching, and they do not capture the full complexity of truly infinite trees.

Yet another approach was proposed by Blumensath in~\cite{DBLP:journals/tcs/Blumensath13a}. Here the algebras are $\omega$-hyperclones, which represent arbitrary infinite trees (with variables) and not just regular ones, thus overcoming  disadvantage (a). Blumensath's notion of ``finite algebra'' is an $\omega$-hyperclone which is finite on every sort (these are infinitely sorted algebraic structures) and which satisfies an additional condition called \emph{path continuity}. A disadvantage of this approach is that path-continuous $\omega$-hyperclones are not closed under homomorphic images, and in particular minimising an algebra can lead outside the class. Furthermore, the notion of path continuity has a similar, if less flagrant, issue  as (b), in the sense that it can be seen as imposing an implicit automaton structure on the algebra.  

In this paper we study $\omega$-clones, which are essentially the same thing as the $\omega$-hyperclones used in~\cite{DBLP:journals/tcs/Blumensath13a}. The difference is cosmetic: we represent single trees instead of tuples of trees. An $\omega$-clone is an infinitely sorted algebraic structure, with sorts indexed by natural numbers:  the $n$-th sort represents possibly infinite trees with $n$ variables, with each variable appearing possibly multiple times (even infinitely often).  

Since $\omega$-clones generalize $\omega$-semigroups in the same way as infinite trees generalize $\omega$-words, they are a natural candidate for an algebra of infinite trees. The question is: what is the right definition of  ``regular'' $\omega$-clone that would correspond to \mso for infinite trees? One would want such a definition to use only structural properties such as finiteness, and not to mention \mso in any way. Ideally, ``regular'' $\omega$-clones should also be closed under basic algebraic operations such as taking homomorphic images, products and finitely generated subalgebras.

In this paper we refute the most natural candidate for such a structural characterisation of regular $\omega$-clones: we construct an $\omega$-clone which is finitely generated, finite on every sort, but it recognizes a tree language that is not regular. This shows that, alone, finiteness  of the universe is insufficient for regularity (unlike in all the algebras for infinite words, or for finite trees). Although this might seem unsurprising (given the infinitely many sorts), simple ideas for a counterexample fail to work and our construction turns out to be rather delicate. This is in contrast with the setting of (finite) graphs, where such counterexamples are easy to find~\cite{courcelle91}.
 

\section{Infinite trees}


A \emph{ranked set} is a set where every element has an associated rank, which is a natural number. For a ranked set $\Sigma$, the set of elements of rank $n$ is denoted $\Sigma_n$. 
A \emph{tree over $\Sigma$} is a tree, possibly with infinite branches or leaves or both, where every node is labelled by an element of $\Sigma$. The number of children for a node must be equal to the rank of its label, and we assume that children are ordered, so that we can speak of the first child, the second child, etc. Here is a picture of a ranked alphabet (with two letters of rank 2 and one of rank 0) and a tree over it.
\begin{center}
	\begin{tabular}{cc}
	a ranked alphabet & a tree over it\\	\imagetop{\includegraphics[page=1,scale=0.2]{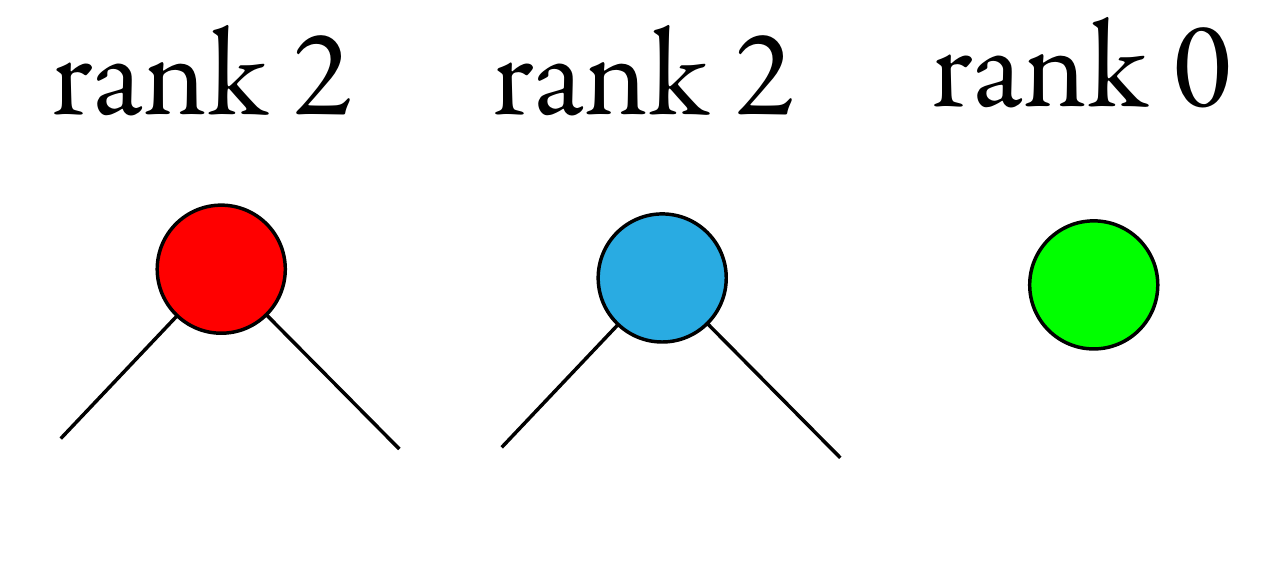}} & 	\imagetop{\includegraphics[page=2,scale=0.2]{pics}}
\end{tabular}
\end{center}
We use standard tree terminology like root, child, parent, leaf, descendant and ancestor. A \emph{branch} in a tree is a maximal (inclusionwise) set of nodes that is totally ordered by the descendant relation. A branch might end in a leaf or be infinite. A \emph{tree language} over a ranked alphabet $\Sigma$ is a set of trees over this alphabet. 

The main focus of this paper is \emph{regular} tree languages, i.e.~those tree languages that can be recognized by automata (equivalently, are \mso definable). The automata we use are nondeterministic parity automata, as defined below.
\begin{defi}\label{def:npa}
	A nondeterministic parity tree automaton consists of:
	\begin{itemize}
		\item an \emph{input alphabet}, which is a finite ranked set $\Sigma$; 
		\item a finite set of \emph{states} $Q$ with a chosen \emph{initial state} $q_0 \in Q$, 
		\item a function that assigns to each state in $Q$ a natural number called its {\em parity value},
		\item for every letter $\sigma \in \Sigma$, a \emph{transition relation}
		\begin{align*}
			\delta_\sigma \subseteq Q \times Q^{\rank \sigma}.
		\end{align*}
	\end{itemize}
\end{defi}
\noindent Since this is the only kind of automata that we use, we will simply call them automata.

A \emph{run} of an automaton over a tree $t$ over the input alphabet $\Sigma$ is a labelling of its nodes by states such that for every node $v$ with label $\sigma \in \Sigma$, say of rank $n$, the transition relation $\delta_\sigma$ contains the tuple consisting of the state in $v$ and the states in the children of $v$, listed from left to right. A run is said to satisfy the {\em parity condition} if on every infinite branch, the maximal parity value seen infinitely often is even. A tree is accepted if there is a run on it which has the initial state in the root and which satisfies the parity condition. The language \emph{recognized} by an automaton is defined to be the set of all trees that it accepts. A tree language is called \emph{regular} if it is recognized by some automaton.

A single infinite tree can also be called {\em regular}: it is when it contains only finitely many non-isomorphic subtrees. Equivalently, a tree is regular if it arises as an unfolding of a finite $\Sigma$-labeled graph, in an obvious sense. It is a standard result due to Rabin (see e.g.~\cite[Thm. 6.18]{Thomas1997}) that a non-empty regular tree language necessarily contains a regular tree. 

Another famous theorem of Rabin is that regular tree languages are exactly those that can be defined in monadic second-order logic (\mso{}), see e.g.~\cite{Thomas1997} for details.

\section{$\omega$-clones}


For a ranked set $\Sigma$, a \emph{term over $\Sigma$ with $n$ ports} is a tree over the disjoint union set $\Sigma + [n]$, where numbers from $[n]=\{1,\ldots,n\}$, called here {\em ports}, are viewed as symbols of rank zero (and can therefore only appear as labels of leaves). We require every port from $1$ to $n$ to appear at least once in such a term, although this is mainly for presentation reasons: it is not difficult to transport our results to a setting where this requirement is not made. Terms may be infinite, and a port may appear infinitely many times in a single term. 

Intuitively, ports in a term are locations where other trees can be grafted to form a larger tree. It is therefore natural to consider a term with $n$ ports as an element of rank $n$ in a ranked alphabet of terms. 

More formally, for a ranked set $\Sigma$, define a ranked set $\clone\Sigma$ so that elements of rank $n$ are terms over $\Sigma$ with $n$ ports, {\em excluding} (for reasons that will become apparent in a moment) the ``trivial'' tree of rank $1$ whose root is labelled with port $1$.

This construction is equipped with the following structure:

\begin{itemize}
	\item  \emph{Action on functions.} A   rank-preserving function $f : \Sigma \to \Gamma$ between ranked sets is lifted to a rank-preserving function $ 
  \clone f : \clone \Sigma \to \clone \Gamma$ by node-wise application of $f$ preserving ports, as in the following example picture:
  
  \begin{center}
	\begin{tabular}{cc}
	$f : \Sigma \to \Gamma$ & $\clone f$ applied to a rank-2 element of $\clone \Sigma$\\	\imagetop{\includegraphics[page=3,scale=0.2]{pics.pdf}} & 	\imagetop{\includegraphics[page=4,scale=0.2]{pics}}
\end{tabular}
\end{center}
\item {\it Unit.} A letter $\sigma \in \Sigma$ of rank $n$ is viewed as a term with $n$ ports which has $\sigma$ in the root and ports in its children, as in the following picture:

 \begin{center}
	\begin{tabular}{cc}
	a rank-3 letter in $\Sigma$ & and its rank-3 unit in $\clone \Sigma$\\	\imagetop{\includegraphics[page=5,scale=0.2]{pics.pdf}} & 	\imagetop{\includegraphics[page=6,scale=0.2]{pics}}
\end{tabular}
\end{center}
This, for any $\Sigma$, defines a rank-preserving unit function $\eta_{\Sigma}:\Sigma\to\clone\Sigma$. 
\item {\it Flattening.} The (rank-preserving) flattening function   $\mu_\Sigma:\clone (\clone \Sigma) \to \clone \Sigma$ is defined as follows.
Suppose that $t \in \clone (\clone \Sigma)$ has rank $n$. Let the root label of $t$ be $s$ (itself a term in $\clone\Sigma$ of rank, say, $k$), and let $t_1,\ldots,t_k$ be the immediate subtrees of the root in $t$, seen as terms in $\clone (\clone \Sigma)$ with rank $n$ each (here, for a brief moment, we allow trees where some ports from $[n]$ may not appear). Then the flattening of $t$ is the term obtained from $s$ by substituting for (each occurrence of) the $i$-th port the flattening of $t_i$, defined recursively, with the special case of $t_i$ being a port (in $t$), which is flattened to itself. (Note that in the flattening of $t$, each of the $n$ ports appears again.) The following picture should make the idea clear: 
 \begin{center}
	\begin{tabular}{cc}
	an element of $\clone (\clone \Sigma)$, of rank 2 & its flattening in $\clone \Sigma$, also of rank 2\\	\imagetop{\includegraphics[page=7,scale=0.2]{pics.pdf}} & 	\imagetop{\includegraphics[page=8,scale=0.2]{pics}}
\end{tabular}
\end{center}
\end{itemize}

\begin{rem}\label{rem:flat}
For a more formal definition it is useful think of nodes in a tree in terms of the unique (finite) branches from the root to those nodes. For a term $t\in\clone(\clone\Sigma)$, a node in the flattening $\mu_{\Sigma}(t)$ is uniquely determined by the following data:
\begin{itemize}
\item a finite branch $(v_1,v_2,v_3,\ldots v_k)$ in $t$ (let $s_i$ denote the term in $\clone\Sigma$ that is the label of $v_i$; for $i=k$ this is defined only if $v_k$ is not a port in $t$),
\item a sequence $(w_1,\ldots,w_{k-1})$, where each $w_i$ is an occurrence in $s_i$ of the port $j$ such that $v_{i+1}$ is the $j$'th child of $v_i$ in $t$, and
\item a node $w_k$ (not a port) in $s_k$ if $v_k$ is not a port; the resulting node in $\mu_{\Sigma}(t)$ is then labelled with the label of $w_k$.
\end{itemize}
If $v_k$ is a port in $t$, then the resulting node in $\mu_{\Sigma}(t)$ is also a port, with the same number.
\end{rem}

\begin{rem}
The trivial term whose root is labelled with port $1$ rather than with an element of $\Sigma$ is excluded from $\clone\Sigma$ because it does not agree well with infinite paths. To see this, for any $\Sigma$, consider a tree which is an infinite path (of rank $0$) whose every node is labelled with the trivial term. One would be in trouble picking a candidate of rank $0$ in $\clone\Sigma$ for the flattening of that infinite path. The same issue arises already in the theory of infinite words, and is the reason why semigroups have more importance than monoids in that theory.
\end{rem}

\begin{rem}
The construction $\clone$ together with its extension to functions, unit, and flattening operations, forms a {\em monad} on the category of ranked sets and rank-preserving functions. Briefly, the unit and flattening operations are natural transformations, flattening is associative, and unit is a two-sided unit for flattening. To keep the exposition elementary we avoid category-theoretic terminology here, but categorical considerations do provide a strong justification for the notion of $\omega$-clones (see below), which are simply Eilenberg-Moore algebras for the monad $\clone$ (see e.g.~\cite[Chap.~VI]{maclane} for the relevant definitions).
\end{rem}

Trees are a natural generalization of words, and the construction $\clone$ is a generalization of the construction of potentially infinite words out of an alphabet. In the same fashion, the following is a generalization of the notion of $\omega$-semigroup.

\begin{defi}
An {\em $\omega$-clone} consists of:
\begin{itemize}
\item a ranked set $A$ equipped with
\item a rank-preserving function $\mult^A:\clone A\to A$ (the {\em product operation}),
\end{itemize}
such that:
\begin{enumerate}
\item for every $a\in A$, the product of the unit of $a$ is $a$ itself:
\[
	\mult^A(\eta_A(a)) = a,
\]
\item for every $t\in \clone(\clone A)$, applying the product operation to the flattening of $t$ yields the same value in $A$ as applying it to the term obtained from $t$ by node-wise application of the product operation:
\[
	\mult^A(\mu_A(t))=\mult^A(\clone\,\mult^A(t)).
\]
\end{enumerate}
\end{defi}
\noindent Intuitively, $\mult^A$ describes a way folding terms labeled with elements of $A$ into single elements of $A$. We call it the {\em product} following the standard terminology for a corresponding operation in $\omega$-semigroups (see e.g.~\cite{perrin_pin_words}).

The first axiom of $\omega$-clones says that products of unit terms are computed trivially; 
the second axiom says that taking the product is compositional with respect to flattening of composite terms.

A {\em homomorphism} from an $\omega$-clone $A$ to $B$ is a rank-preserving function $h:A\to B$ that commutes with the respective product operations in the expected sense:
\[
	h(\mult^A(t)) = \mult^B(\clone h(t)) \quad \mbox{for every }t\in\clone A.
\]

It easily follows from the definition that for any ranked set $\Sigma$, the set $\clone \Sigma$ is an $\omega$-clone, with the flattening taken as the product operation. Indeed, it is the {\em free $\omega$-clone} over $\Sigma$: for any $\omega$-clone $A$, every rank-preserving function from $\Sigma$ to $A$ extends uniquely to an $\omega$-clone homomorphism from $\clone \Sigma$ to $A$.

A tree language $L$ over a ranked alphabet $\Sigma$ can be seen as a subset of $\clone\Sigma$ that only contains terms of rank $0$. We say that an $\omega$-clone $A$ {\em recognizes} $L$ if there is a homomorphism $h:\clone \Sigma\to A$, and a subset $B\subseteq A$, such that the inverse image of $B$ along $f$ is exactly $L$. 

Inspired by similar results regarding regular languages of words, infinite words or finite trees (see Introduction), one would ideally want that a tree language over a finite alphabet is regular if and only if it is recognized by a finite $\omega$-clone. This turns out to be false very quickly. Indeed, if a ranked set $A$ has some element of rank more than $1$, then the set $\clone A$ has elements of arbitrary finite rank. Since the product operation in an $\omega$-clone must be rank-preserving, for $A$ to be an $\omega$-clone it must be non-empty on every rank, so it cannot be finite. As a result, all finite $\omega$-clones have only elements of rank at most $1$, so they are essentially $\omega$-semigroups and they cannot recognize anything beyond languages of infinite words. 

One must therefore relax the finiteness restriction on $\omega$-clones to hope for a correspondence to regular tree languages. A natural idea is to require a clone to be {\em rank-wise finite} (i.e. finite on every rank), and finitely generated. An $\omega$-clone $A$ is {\em finitely generated} if there is a finite subset $G\subseteq A$ such that every element of $A$ can be obtained as the product of some term in $\clone G$. 

With this relaxed definition, one direction of the desired correspondence holds:

\begin{thm}\label{thm:regular-algebra-languages}
Every regular tree language over a finite alphabet is recognized by a rank-wise finite, finitely generated $\omega$-clone.
\end{thm}
\begin{proof}[Proof (sketch)]
The $\omega$-clone to recognize a regular tree language is built of ``profiles'' of automata runs. Consider runs of an automaton $\mathcal{A}$ over input alphabet $\Sigma$ on a $\Sigma$-tree $t$ with $n$ ports. Restrict attention to those runs that satisfy the parity condition, i.e. those where on every infinite branch the maximal parity value seen infinitely often is even. The {\em profile} of such a run consists of:
\begin{itemize}
\item the initial state, i.e., the state at the root of $t$,
\item for each (occurrence of a) port $i\in[n]$ in $t$, a triple $(q',m,i)$, where $q'$ is the state at that occurrence (which is a leaf in $t$), and $m$ is the maximal parity value at the (finite) branch leading to that leaf.
\end{itemize}
For each automaton $\mathcal{A}$, there are finitely many possible profiles of every rank $n$. The set of sets of those profiles can be equipped with an $\omega$-clone structure (memoryless determinacy of parity games is used to ensure that it is indeed an $\omega$-clone). This $\omega$-clone recognizes the language accepted by $\mathcal{A}$.
\end{proof}

A full proof of Theorem~\ref{thm:regular-algebra-languages} is rather technical. We omit the details, however, because our main message in this paper is that the converse implication fails. The remainder of this paper is devoted to defining a tree language, over a finite alphabet, that is not regular but is nevertheless recognized by a rank-wise finite, finitely generated $\omega$-clone.


\section{Densely antiregular trees}


Fix a ranked alphabet $\set{a,b}$ with two letters, both of them of rank 2. Since this alphabet has no letters of rank 0, every tree (without ports) over it is a full binary tree. Call such a tree \emph{antiregular} if every two different nodes have different subtrees. 

\begin{lem}\label{lem:exist-antiregular}
Antiregular trees exist. 
\end{lem}
\begin{proof}
A node in the full binary tree can be viewed as a word in $\set{0,1}^*$, with $0$ indicating a left turn and $1$ indicating a right turn. A tree over the alphabet $\set{a,b}$ can be viewed as a  subset of $\set{0,1}^*$, i.e.~a language of finite words, which contains those words which correspond to nodes labeled with $a$. It is not difficult to show that a tree is antiregular if and only if in the corresponding language of finite words, the Myhill-Nerode equivalence relation has only singleton equivalence classes. Languages with the latter property exist, e.g.~the language of palindromes.
\end{proof}

The set of antiregular trees itself is not recognized by any rank-wise finite $\omega$-clone. Intuitively, in order to determine whether a tree is antiregular, assuming that both subtrees $t_0,t_1$ of the root are antiregular, one needs to store an infinite amount of information, namely all subtrees of  $t_0$ and all subtrees of $t_1$. We shall therefore relax the antiregularity condition now.

A set $V$ of nodes in a tree is called \emph{dense} if every node of the tree has a descendant (not necessarily proper) in $V$.  Call a tree over $\set{a,b}$ \emph{densely antiregular} if the set $\set{v : \mbox{$v$ is a node whose subtree is antiregular}}$ is dense.

\begin{lem}\label{lem:densely-antiregular-is-not-regular}
The language of densely antiregular trees is not regular.
\end{lem}
\begin{proof}
Every antiregular tree is densely antiregular, so by Lemma~\ref{lem:exist-antiregular} the language of densely antiregular trees is not empty. Furthermore, a regular tree cannot be densely antiregular. The lemma follows from the the fact that every non-empty, regular tree language contains a regular tree (see e.g.~\cite[Thm. 6.18]{Thomas1997}). 
\end{proof}

We will show that the set of densely antiregular trees is recognized by an $\omega$-clone $A$ that is finite on every rank and finitely generated. To define $A$, we use a classification of terms into four kinds. Consider  a term in $ \clone \set{a,b}$ of rank $n$. Such a term is a tree which uses labels $a,b$ in non-leaf nodes, and ports $1,\ldots,n$ in leaves. A term can be of one of four mutually exclusive kinds:
  \begin{itemize}
  	\item[{\bf 1}:] some subtree has no ports and is not densely antiregular;
  	\item[{\bf 2}:] some subtree has no ports and every subtree without ports is densely antiregular;
  	\item[{\bf 3}:] every subtree has a port and some port name is used at least twice;
  	\item[{\bf 4}:] every subtree has a port and every port name is used exactly once.
  \end{itemize}
  
For illustration, here are examples of terms of the four kinds:  
  
 \begin{center}
	\begin{tabular}{cc}
	kind {\bf 1}: & kind {\bf 2}:\\	
\imagetop{\includegraphics[page=15,scale=0.2]{pics.pdf}} & 	\imagetop{\includegraphics[page=16,scale=0.2]{pics}} \\ \\
	kind {\bf 3}: & kind {\bf 4}:\\	
\imagetop{\includegraphics[page=17,scale=0.2]{pics.pdf}} & 	\imagetop{\includegraphics[page=18,scale=0.2]{pics}}
\end{tabular}
\end{center}  
  
\noindent Note that a term of kind {\bf 4} must be finite and its number of leaves is equal to its rank; as a consequence, for every rank there are only finitely many terms of kind {\bf 4}.

Define a ranked set $A$ so that its element of rank $n$ is:
\begin{itemize} 
\item a term of kind {\bf 4} of rank $n$, or 
\item a kind name ({\bf 1}, {\bf 2} or {\bf 3}) together with the number $n$ (the rank); kind name {\bf 3} is included only if $n$ is positive.
\end{itemize}
By the above remark, $A$ is finite on every rank. 

There is an easy (surjective) rank-preserving function $h:\clone\set{a,b}\to A$ that maps terms of kind ${\bf 4}$ identically, and maps every other term to its kind and rank. Note that all trees of rank $0$ are of kind {\bf 1} or {\bf 2}. The function $h$ maps densely antiregular trees to {\bf 2}, and all the other ones to {\bf 1}. In particular, $h$ recognizes the set of densely antiregular trees. 

We will build an $\omega$-clone structure on the set $A$ that will make $h$ a homomorphism. To this end we will show that the value of $h$ on the flattening of a term $t\in\clone(\clone\{a,b\})$ can be derived just from the external structure of $t$ and from the values of $h$ on its nodes.

Here and in the following, $t|_v$ denotes the subtree of $t$ rooted at a node $v$. 

\begin{lem}\label{lem:flattening-subtree}
For every subtree $s$ of the flattening of $t\in\clone(\clone\{a,b\})$, one of the following conditions holds:
\begin{itemize}
	\item[(A)] some label of a node in $t$ contains $s$ as a subtree; or
	\item[(B)] for some node $v$ of $t$, the flattening of $t|_v$ is a proper subtree of $s$.
\end{itemize}
\end{lem}
\begin{proof}
Easy from the definition of flattening. Looking at the formulation in Remark~\ref{rem:flat}, if the root of $s$ is determined by sequences
\[
	(v_1,v_2,\ldots,v_k) \qquad \text{and} \qquad (w_1,w_2,\ldots,w_k)
\]
and $s_k$ is the label of $v_k$, then case $(A)$ holds if the subtree of $s_k$ rooted at $w_k$ contains no ports, since then $s$ is identical to that subtree. On the other hand if that subtree contains a port, say of number $j$, then the flattening of the $j$'th subtree of $v_k$ in $t$ is a proper subtree of $s$ and case (B) holds.
\end{proof}

Terms in $\clone\{a,b\}$ cannot have any leaves other than ports, so the terms ``leaf'' and ``port'' can be used interchangeably for them. However, this does not apply to terms $t\in\clone(\clone\{a,b\})$, as such a term can have a leaf labeled with a tree in $\clone\{a,b\}$ of rank $0$. Reading the following discussion, one should keep this distinction in mind.

\begin{lem}\label{lem:has-a-leaf}
Every subtree of the flattening of $t$ has a port if and only if:
\begin{itemize}
\item every subtree of $t$ has a port, and
\item all nodes in $t$ are of kind {\bf 3} or {\bf 4}.
\end{itemize}
\end{lem}
\begin{proof}
If $t$ has a node of kind {\bf 1} or {\bf 2}, then the label of that node has a subtree with no ports, and that subtree persists in the flattening of $t$. If $t$ itself has a subtree with no ports, then the flattening of that subtree has no ports. 

Conversely, assume that the flattening of $t$ contains a subtree $s$ with no ports. In the case (A) from Lemma~\ref{lem:flattening-subtree}, $s$ was a subtree of the label of a node in $t$, and that node must have been of kind ${\bf 1}$ or ${\bf 2}$. In the case (B), there is a node $v$ of $t$ such that the flattening of $t|_v$ has no ports. This can happen only if $t|_v$ has no ports.
\end{proof}

\begin{lem}\label{lem:port-twice}
Assume that every subtree of the flattening of $t$ has a port. Some port name is used in the flattening of $t$ at least twice if and only if:
\begin{itemize}
\item some port occurs in $t$ at least twice, or
\item some node in $t$ is of kind {\bf 3}.
\end{itemize}
\end{lem}
\begin{proof}
By Lemma~\ref{lem:has-a-leaf}, every subtree of $t$ has a port and all nodes in $t$ are of kind {\bf 3} or {\bf 4}.

If some port occurs in $t$ at least twice, then the port occurs at least twice in the flattening of $t$ too. On the other hand, if some node $v$ in $t$ is of kind {\bf 3}, then pick a port name $k$ that is used at least twice in the label of $v$. The subtree of $t$ rooted at the $k$-th child of $v$ has some
port $w$.  
That port appears at least twice in the flattening of $t$.

Conversely, if every port name in $t$ is used only once and every node in $t$ is of type {\bf 4}, then it is easy to see that every port name in the flattening of $t$ is also used only once.
\end{proof}

\begin{lem}\label{lem:antiregular}
Assume that $t$ has no ports. Then the flattening of $t$ is densely antiregular if and only if
\begin{enumerate}
	\item[(i)] no node in $t$ is of kind {\bf 1}; and
	\item[(ii)] every node in $t$ has a descendant $v$ such that:
	\begin{enumerate}
		\item the label of $v$ is of kind {\bf 2}; or
		\item all nodes in $t|_v$ are of kind {\bf 4} and the flattening of $t|_v$ is densely antiregular.
	\end{enumerate}
\end{enumerate}
\end{lem}
\begin{proof}
Let us first show the ``if'' part. Assume that $t \in \clone (\clone \{a,b\})$ has no ports and satisfies conditions (i) and (ii). We need to show that every subtree $s$ of the flattening of $t$ has a subtree that is antiregular. Notice that the flattening of $t$ has no ports; as a consequence, $s$ has no ports either.

Consider first the case (A) from Lemma~\ref{lem:flattening-subtree}, i.e.~that $s$ is a subtree of some label $r \in \clone\{a,b\}$ of a node in $t$. In particular the label $r$ has a subtree without ports, and therefore $r$ must be of kind {\bf 1} or {\bf 2}. Assumption (i) rules out kind {\bf 1}, and therefore $r$ has kind {\bf 2}, which implies that $s$ is densely antiregular, and thus has an antiregular subtree.  We are left with case (B), where some subtree of $s$ is equal to the flattening of $t|_w$ for some node $w$ in $t$. By assumption (ii), $w$ has a descendant $v$ in $t$ that satisfies one of (a) or (b) as in the statement of the lemma. In either case the flattening of $t|_v$ (and hence also $s$) contains an antiregular subtree.

For the ``only if'' part, suppose that the flattening of $t \in \clone (\clone \{a,b\})$ is densely antiregular. Clearly no node in $t$ can have kind {\bf 1}, since every subtree of a densely antiregular tree is also densely antiregular. It remains to prove condition (ii). Take some node $u$ in $t$. If $u$ has a descendant with label of kind {\bf 2} then (ii) is satisfied. Suppose otherwise. Assume first that $u$ has some descendant $v$ such that the subtree $t|_v$ uses only labels of kind {\bf 4}. Since the flattening of $t|_v$ is a subtree of the flattening of $t$, it must be densely antiregular, and therefore (ii) holds thanks to (b). The remaining case is when no such $v$ exists, hence nodes of kind {\bf 3} are dense in the subtree $t|_u$. But then every subtree of the flattening of $t|_u$ contains two identical subtrees and is therefore not antiregular, which contradicts the assumption that the flattening of $t$ is densely antiregular.
\end{proof}

\begin{cor}\label{cor:classification}
For any term $t\in\clone(\clone\{a,b\})$, the following properties hold:
\begin{itemize}
\item[(a)] If some node of $t$ is of kind {\bf 1} then the flattening of $t$ is of kind {\bf 1}.
\item[(b)] If some subtree of $t$ has no 
ports and does not satisfy condition (ii) in Lemma~\ref{lem:antiregular}, then the flattening of $t$ is of kind {\bf 1}.
\item[(c)] If:
\begin{itemize}
\item no node of $t$ is of kind {\bf 1},
\item some subtree of $t$ has no
ports, and
\item every subtree of $t$ with no
ports satisfies condition (ii) in Lemma~\ref{lem:antiregular},
\end{itemize}
then the flattening of $t$ is of kind {\bf 2}.
\item[(d)] If:
\begin{itemize}
\item no node of $t$ is of kind {\bf 1},
\item every subtree of $t$ has
ports, and
\item some node of $t$ is of kind {\bf 2},
\end{itemize}
then the flattening of $t$ is of kind {\bf 2}.
\item[(e)] If:
\begin{itemize}
\item all nodes of $t$ are of kind {\bf 3} or {\bf 4},
\item every subtree of $t$ has ports, and
\item some node of $t$ is of kind {\bf 3},
\end{itemize}
then the flattening of $t$ is of kind {\bf 3}.
\item[(f)] If:
\begin{itemize}
\item all nodes of $t$ are of kind {\bf 4},
\item every subtree of $t$ has ports, and
\item some port name appears in $t$ more than once,
\end{itemize}
then the flattening of $t$ is of kind {\bf 3}.
\item[(g)] If:
\begin{itemize}
\item all nodes of $t$ are of kind {\bf 4},
\item every subtree of $t$ has ports, and
\item every port name appears in $t$ exactly once,
\end{itemize}
then the flattening of $t$ is of kind {\bf 4}.
\end{itemize}
Moreover, cases (a)-(g) cover all terms $t$ in $\clone(\clone\{a,b\})$. 
\end{cor}
\begin{proof}
The final remark is easy to check. 
\begin{itemize}
\item[(a)] If the label of a node $v$ in $t$ is of kind {\bf 1} then it contains a subtree which has no ports and it not densely antiregular. That subtree appears in the flattening of $t$, therefore the flattening of $t$ is of kind {\bf 1}.
\item[(b)] Let $v$ be a node in $t$ such that $t|_v$ has no
ports and does not satisfy condition (ii) in Lemma~\ref{lem:antiregular}. Then the flattening of $t|_v$ has no ports and, by Lemma~\ref{lem:antiregular}, it is not densely antiregular. As a result, the flattening of $t$ is of kind {\bf 1}.
\item[(c)] Let $v$ be a node of $t$ such that $t|_v$ has no
ports. Such a node exists by our second assumption. Clearly, the flattening of $t|_v$ has no ports.

Now consider any subtree $s$ of the flattening of $t$ that has no ports. We need to prove that $s$ is densely antiregular, i.e., that every subtree $s'$ of $s$ contains an antiregular subtree. Apply Lemma~\ref{lem:flattening-subtree} to $s'$. In case (A), $s'$ is a subtree of (the label of) a node $v$ of $t$. Since $t$ has no nodes of kind {\bf 1}, $v$ must be of kind {\bf 2}, hence $s'$ is densely antiregular. In case (B), there is a node $v$ in $t$ such that the flattening of $t|_v$ is a subtree of $s'$. Since $s'$ has no ports, $t|_v$ cannot have ports either so, by our assumptions, $t|_v$ satisfies condition (ii) in Lemma~\ref{lem:antiregular}. Since $t$ has no nodes of kind {\bf 1} condition (i) is also satisfied, hence by Lemma~\ref{lem:antiregular} the flattening of $t|_v$ is densely antiregular so it contains an antiregular subtree as required.

As a result, the flattening of $t$ is of kind {\bf 2}.
\item[(d)] If some node $v$ of $t$ is of kind {\bf 2}, then its label has a subtree with no ports, and that subtree persists in the flattening of $t$.

Now consider any subtree $s$ of the flattening of $t$ such that $s$ has no ports. We need to prove that $s$ is densely antiregular. Apply Lemma~\ref{lem:flattening-subtree} to $s$. In case (A), $s$ is a subtree of a node $v$ of $t$. Since $t$ has no nodes of kind {\bf 1}, $v$ must be of kind {\bf 2}, hence $s$ is densely antiregular. In case (B), there must be a node $v$ in $t$ such that the flattening of $t|_v$ is a subtree of $s$. But by our assumptions $t|_v$ has ports, so the flattening of $t|_v$ has ports, which is a contradiction. As a result, case (B) cannot happen, $s$ is densely antiregular and the flattening of $t$ is of kind {\bf 2}.
\item[(e)] By the first two assumptions and by Lemma~\ref{lem:has-a-leaf}, every subtree of the flattening of $t$ has a port. Furthermore, since $t$ has a node of kind {\bf 3}, by Lemma~\ref{lem:port-twice} the flattening of $t$ is of kind {\bf 3}.
\item[(f)] Proved similarly to (e).
\item[(g)] By the first two assumptions and by Lemma~\ref{lem:has-a-leaf}, every subtree of the flattening of $t$ has a port. Since $t$ has no node of kind {\bf 3} and no port is used more than once in $t$, by Lemma~\ref{lem:port-twice} the flattening of $t$ is of kind {\bf 4}. \qedhere
\end{itemize}
\end{proof}

\noindent The rather tiresome Corollary~\ref{cor:classification} shows how, given a tree $t\in\clone(\clone\{a,b\})$, one can determine the kind of the flattening of $t$ just by looking at the external structure of $t$ and at the kinds of (the labels of) its nodes. More precisely, it is enough to look at the tree $\clone h(t)\in\clone A$. The procedure can be described by a decision diagram:

\[\xymatrix{
& *+[F]{\txt{Some node \\ is of kind {\bf 1}?}}\ar`l[ld]_-{\text{yes}}[ldddd]_>>{\text{(a)}}\ar[d]_{\text{no}} \\
& *+[F]{\txt{Every subtree \\ has a port?}}\ar[r]^-{\text{yes}}\ar[d]_-{\text{no}} & *+[F]{\txt{Some node \\is of kind {\bf 2}?}}\ar@(lu,r)[lddd]_(.35){\text{yes}}^>>{\text{(d)}}\ar[d]^{\text{no}} \\
& *+[F]{\txt{Every subtree \\ with no ports \\ satisfies condition (ii) \\ in Lemma~\ref{lem:antiregular}?}}\ar[ldd]_(.4){\text{no}}^>>{\text{(b)}}\ar[dd]^(.4){\text{yes}}_>>{\text{(c)}}  & *+[F]{\txt{Some node \\is of kind {\bf 3}?}}\ar[d]^{\text{no}}  \\
& & *+[F]{\txt{Some port name appears \\ more than once?}}\ar[d]^-{\text{no}}_>>{\text{(g)}} \\
 *+[o][F]{\bf 1} &  *+[o][F]{\bf 2} \save[]+<2.3cm,0cm>*+[o][F]{\bf 3}="kind3"\restore &  *+[o][F]{\bf 4} 
 \ar@/_3pc/[uu];"kind3"^(.4){\text{yes}}_>>{\text{(e)}} 
 \ar[u];"kind3"^{\text{yes}}^>>{\text{(f)}}
}\]

Final decisions in the diagram are labeled with respective subcases of Corollary~\ref{cor:classification} that justify them.

Note that the only path that leads to outcome {\bf 4} guarantees that all nodes (apart from ports) in $t$ are of kind {\bf 4}. In this case, the tree $\clone h(t)$ retains full information about all nodes. This means that $\clone h(t)$ contains enough information to determine not just the {\em kind} of the flattening of $t$, but also the value of $h$ on that flattened tree. Formally, this defines a function, which we denote $\mult^A:\clone A\to A$, such that the following diagram commutes:
\[\xymatrix{
\clone(\clone\{a,b\})\ar[r]^-{\clone h}\ar[d]_{\mu_{\{a,b\}}} & \clone A\ar[d]^{\mult^A} \\
\clone\{a,b\}\ar[r]_h & A.
}\]
This almost means that $h$ is a homomorphism; the only thing that remains to be proved is that $\mult^A$ is an $\omega$-clone on $A$. This would be tedious to check by hand, but fortunately it follows from abstract considerations: for every commuting diagram of functions
\[\xymatrix{
\clone B\ar[r]^{\clone h}\ar[d]_f & \clone A\ar[d]^g \\
B\ar[r]_h & A,
}\]
if $f$ is an $\omega$-clone on $B$ and $h$ is surjective then $g$ is an $\omega$-clone on $A$. This holds for Eilenberg-Moore algebras for any monad in place of $\clone$, and it is a folklore result proved in passing by several authors, see e.g~\cite[p. 152]{maclane} or~\cite[p. 96]{ttt}. A more explicit statement and proof can be found in~\cite[Lem.~3.3]{DBLP:journals/corr/Bojanczyk15}.

Finally, it is easy to see that the $\omega$-clone $A$ is finitely generated by the set of all its elements of rank at most $2$.

This completes the proof that the language of densely antiregular trees, although not regular, is recognized by a rank-wise finite, finitely generated $\omega$-clone.

\bibliographystyle{alpha}
\bibliography{bib}

\end{document}